\documentclass[conference]{IEEEtran}
\IEEEoverridecommandlockouts
\usepackage{cite}

\usepackage{amsfonts,amsmath,amssymb,amsthm,mathrsfs,bm,bbm}
\usepackage{graphicx}
\usepackage{overpic}
\usepackage{textcomp}
\usepackage{xcolor}
\usepackage{flushend}
\usepackage{enumerate}
\usepackage{algpseudocode}

\ifCLASSOPTIONcompsoc 
	\usepackage[caption=false,font=normalsize,labelfon t=sf,textfont=sf]{subfig} 
\else 
	\usepackage[caption=false,font=footnotesize]{subfig} \fi 
    
\newcommand{\eqdef}{:=}

\renewcommand{\vec}[1]{\bm{#1}}		
\newcommand{\rvec}[1]{\mathbbm{#1}} 		
\newcommand{\rmat}[1]{\mathbbm{#1}} 	
\newcommand{\E}{\mathsf{E}}		
\newcommand{\Var}{\mathsf{V}}			
\newcommand{\stdset}[1]{\mathbbmss{#1}}	
\newcommand{\set}[1]{\mathcal{#1}}		

\newcommand{\CN}{\mathcal{CN}}			
\newcommand{\herm}{\mathsf{H}}			
\newcommand{\T}{\mathsf{T}}				

\newtheorem{definition}{Definition}

\newtheorem{proposition}{Proposition}

\newtheorem{remark}{Remark}

\IEEEoverridecommandlockouts\IEEEpubid{\makebox[\columnwidth]{ 978-1-6654-3540-6/22~\copyright~2022 IEEE \hfill} \hspace{\columnsep}\makebox[\columnwidth]{ }}

\begin{document}

\title{Joint optimal beamforming and power control in cell-free massive MIMO
\thanks{L. Miretti, R. L. G. Cavalcante, and S. Sta\'nczak acknowledge the financial support by the Federal Ministry of Education and Research of Germany in the programme of “Souverän. Digital. Vernetzt.” Joint project 6G-RIC, project identification number: 16KISK020K and 16KISK030.}
}

\author{
\IEEEauthorblockN{Lorenzo Miretti, Renato L. G. Cavalcante, Sławomir Sta\'nczak}
\IEEEauthorblockA{\emph{Fraunhofer Institute for Telecommunications Heinrich-Hertz-Institut} \\
\emph{Technical University of Berlin}\\
Berlin, Germany \\
\{lorenzo.miretti,renato.cavalcante,slawomir.stanczak\}@hhi.fraunhofer.de}
}

\maketitle

\begin{abstract} We derive a fast and optimal algorithm for solving practical weighted max-min SINR problems in cell-free massive MIMO networks. For the first time, the optimization problem jointly covers long-term power control and distributed beamforming design under imperfect cooperation. In particular, we consider user-centric clusters of access points cooperating on the basis of possibly limited channel state information sharing. Our optimal algorithm merges powerful power control tools based on interference calculus with the recently developed \textit{team} theoretic framework for distributed beamforming design. In addition, we propose a variation that shows faster convergence in practice.
\end{abstract}

\begin{IEEEkeywords}
team MMSE, cell-free massive MIMO, power control, distributed beamforming, user-centric
\end{IEEEkeywords}

\section{Introduction} 
Cell-free massive multiple-input multiple-output (MIMO) is one of the key research avenues for future generation sub-$6$GHz wireless access networks. Its main focus is the study of simple and scalable access-point (AP) cooperation architectures in ultra-dense networks, able to offer uniformly good service in a cost-effective manner \cite{ngo2017cell,bjornson2020scalable,buzzi2020,demir2021}. 

To achieve this ambitious goal, cell-free massive MIMO has evolved from related standardized technologies, such as 4G and 5G coordinated multi-point (CoMP), by introducing flexible cooperation regimes with lower information exchange and joint processing requirements. A first innovation is the adoption of \emph{user-centric} designs, where multiple network components are grouped into possibly overlapping cooperation clusters, each of them allocated to a single user equipment (UE), without relying on the notion of a cell \cite{bjornson2020scalable,buzzi2020}.  A second crucial aspect is the optimal distribution of signal processing tasks between remote control units and APs such that information exchange is minimized. For instance, in the original cell-free massive MIMO idea \cite{ngo2017cell}, all short-term array processing tasks, such as channel estimation and beamforming, are performed by the APs on the basis of only \emph{local} channel measurements. Then, joint encoding/decoding and resource allocation tasks such as power control are performed at a remote processing unit by only assuming knowledge of long-term statistical information. Many variations of this cooperation regime are currently subject of active research \cite{bjornson2019making,attarifar2020subset,shaik2020mmse}. 

Against this background, this study considers the important open problem of joint optimal long-term power control and distributed beamforming design under imperfect channel state information (CSI) sharing across the APs. To relate the present study with existing approaches, we recall that iterative algorithms for long-term power control under fixed beamforming design have been extensively studied in both the cellular and cell-free massive MIMO literature (see, e.g., \cite{ngo2017cell,demir2021,marzetta2016fundamentals, massivemimobook}). In contrast, the problem of optimal distributed beamforming design without power control has been solved only recently in \cite{miretti2021team,miretti2021team2} using elements of \emph{team} decision theory, a rigorous control theoretical framework for decentralized decision making under asymmetry of information. 
However, none of these previous studies have jointly optimized beamformers and transmit powers. The closest approaches to the problem at hand date back to the studies on joint beamforming design and power control in \textit{static} channels (see, e.g., \cite{schubert2004solution,schubert2011interference}). 
Although conceptually similar, these methods only apply to centralized or purely statistical beamforming. Furthermore, the centralized case typically considers rather impractical short-term optimization, i.e., for each channel state realization. 

Unlike previous studies, we target a certain class of weighted max-min SINR problems where, for the first time, joint optimal long-term power control and distributed beamforming design is performed. 
By tuning the problem weights, the corresponding solutions attain all positive boundary points  of the achievable rate region predicted by the popular \emph{use-and-then-forget} (UatF) ergodic capacity inner bound \cite{marzetta2016fundamentals}, which is particularly suitable for distributed cell-free operations~\cite{demir2021}. We propose two fast algorithms: for the first algorithm, global optimality is theoretically demonstrated; for the second algorithm, convergence to the global optimum is illustrated via simulations. On top of handling distributed cell-free operations, the proposed algorithms may also offer a more practical alternative to the available approaches for centralized operations. 

\subsection{Notation and mathematical preliminaries}
Let $(\Omega,\Sigma,\mathbb{P})$ be a probability space. We denote by 
$\set{H}^K$ the space of $K$-dimensional random vectors $\Omega \to \stdset{C}^K$, i.e., $K$-tuples of $\Sigma$-measurable functions $\Omega \to \stdset{C}$, satisfying $(\forall \rvec{h}\in \set{H}^K)$ $\E[\|\rvec{h}\|_2^2]<\infty$. 
We denote by $\stdset{R}_+$ and $\stdset{R}_{++}$ the sets of, respectively, nonnegative and positive reals. The $k$th column of the $K$-dimensional identity matrix $\vec{I}_K$ is denoted by $\vec{e}_k$. Inequalities involving vectors should be understood coordinate-wise. A norm $\|\cdot\|$ on $\stdset{R}^K$ is said to be \emph{monotone}  if $(\forall\vec{x}\in\stdset{R}_+^K)(\forall\vec{y}\in\stdset{R}_+^K)~\vec{x}\le\vec{y}\Rightarrow \|\vec{x}\|\le\|\vec{y}\|$.

\begin{definition}\label{def:pareto}
Given a nonempty set $\set{R}\subset \stdset{R}_+^K$, we say that $\vec{r}^\star \in \set{R}$ is (weakly) Pareto optimal iff $\nexists \vec{r}\in \set{R}$ such that $\vec{r} > \vec{r}^\star$. The set $\partial \set{R}$ of all (weakly) Pareto optimal points is called the (weak) Pareto boundary of $\set{R}$ \cite{emil2013optimal}.
\end{definition}

\begin{definition}\label{def:interf_fun}
A function $T:\stdset{R}^K_{+}\to\stdset{R}^K_{++}$ is called a \emph{standard interference mapping} if the following properties hold \cite{yates95}: 
\begin{enumerate}[(i)]
	\item $(\forall \vec{x}\in\stdset{R}^K_{+})(\forall \vec{y}\in\stdset{R}^K_{+})~ \vec{x}\geq\vec{y} \implies T(\vec{x})\geq T(\vec{y})$;
	\item $(\forall \vec{x}\in\stdset{R}^K_{+})$ $(\forall \alpha>1)$  $\alpha T(\vec{x})>T(\alpha\vec{x})$.
\end{enumerate}
For later reference, given a standard interference mapping, we call each coordinate function of the mapping a \emph{standard interference function}. For simplicity, we also require continuity for a function to be called a standard interference mapping. 
\end{definition}

\section{Model and problem statement}
\subsection{Capacity of distributed cell-free uplink networks}
Let us consider the uplink of a cell-free network composed of $L$ APs indexed by $\mathcal{L}:=\{1,\ldots,L\}$, each of them equipped with $N$ antennas, and $K$ single-antenna UEs indexed by $\mathcal{K}:=\{1,\ldots,K\}$.  By assuming conventional flat-fading synchronous channel models \cite{tse2005fundamentals}, we focus on the ergodic capacity in the classical Shannon sense, which is approximated with the popular \emph{use-and-then-forget} (UatF) inner bound \cite{marzetta2016fundamentals}. In more detail, we define $(\forall k \in \set{K})$
\begin{equation*}
R_k(\rvec{v}_k,\vec{p}) \eqdef \log_2(1+\mathrm{SINR}_k(\rvec{v}_k,\vec{p})),
\end{equation*}
\begin{equation*}
\mathrm{SINR}_k(\rvec{v}_k,\vec{p}) \eqdef \resizebox{0.69\linewidth}{!}{$\dfrac{p_k|\E[\rvec{h}_k^\herm\rvec{v}_k]|^2}{p_k\Var(\rvec{h}_k^\herm\rvec{v}_k)+\underset{j\neq k}{\sum} p_j\E[|\rvec{h}_j^\herm\rvec{v}_k|^2]+\E[\|\rvec{v}_k\|_2^2]}$},
\end{equation*}
where $\vec{p} \eqdef [p_1,\ldots,p_K]^\T \in \stdset{R}_{+}^K$ is a deterministic vector of transmit powers, $\rvec{h}_k \in\set{H}^{NL}$ is a random channel vector modeling the fading state between UE $k$ and all APs, and $\rvec{v}_k \in \set{H}^{NL}$ is a random beamforming vector, which is applied by the APs to process jointly the signal of UE~$k$. We assume that $\rvec{v}_k\neq\vec{0}$ almost surely (a.s.), so  $\E[\|\rvec{v}_k\|_2^2]\neq 0$ and the SINR is well defined. We stress that \emph{random} beamforming here simply refers to the fact that beamformers may adapt to fading realizations. 
The achievable rate region is given by
\begin{equation*}
\set{R} \eqdef \underset{\substack{\|\vec{p}\|_{\infty}\leq P \\ (\forall k \in \set{K})~\rvec{v}_k \in \set{V}_k \\(\forall k \in \set{K})~\E[\|\rvec{v}_k\|_2^2] \neq 0}}{\bigcup} \left\{ \vec{r}\in \stdset{R}_+^K \middle| \begin{array}{c}
r_1 \leq R_1(\rvec{v}_1,\vec{p}) \\
\vdots \\
r_K \leq R_K(\rvec{v}_K,\vec{p})
\end{array} \right\},
\end{equation*} 
where $0<P < \infty$ denotes the per-UE power budget, and where the set $\set{V}_k\subseteq \set{H}^{NL}$ denotes an \emph{information} constraint, defined in the following section, modelling \emph{ distributed} cell-free operations such as in \cite{bjornson2020scalable}. The region $\set{R}$ is especially relevant if codewords can span many realizations of a stationary and ergodic fading process, and it relies on the following widely-used information theoretical simplifications \cite{marzetta2016fundamentals,caire2018ergodic}:
\begin{enumerate}[(i)]
\item Correlation across different fading realizations is neglected in the encoding and decoding phase, although it can be exploited in the channel estimation phase (the overhead of this operation is left unspecified here);
\item CSI is neglected in the decoding phase, but it is exploited for coherent receive beamforming, i.e., to generate the soft estimates feeded as an input to the channel decoder;
\item \emph{Treating-interference-as-noise} (TIN) is considered in the decoding phase. More precisely, the message of UE~$k$ is transmitted using simple coding techniques for the memoryless point-to-point Gaussian channel, while interference is mitigated by only using spatial processing.
\end{enumerate}

\begin{remark}
	\label{remark.practical_aspects} 
Alternative rate regions based on TIN with various forms of CSI at the decoder are often considered (see  \cite{massivemimobook,caire2018ergodic,gottsch2022subspace} and references therein). These regions may provide  less pessimistic capacity approximations, but optimization problems  become challenging because of the expectation outside the logarithm in the resulting rate expressions. To deal with this issue, most studies perform short-term  optimization, i.e., an optimization problem is solved for each fading realization. However, 
this approach is hardly practical in cell-free networks because of the 
overhead required to exchange instantaneous CSI among many geographically distant network elements. In contrast, the proposed algorithms based on $\set{R}$ are readily implementable in distributed cell-free networks such as the ones envisioned in \cite{demir2021}. In these networks, short-term array processing is confined to the APs, while joint decoding and  optimization (e.g., power control) takes place at a remote processing unit based on long-term statistical information.
\end{remark}

\subsection{Information constraints}
In distributed cell-free operation, APs joint processing capabilities are impaired. One of the major impairment is that the APs must compute their portion of the beamformers on the basis of imperfect CSI sharing, or even based on only \emph{local} CSI \cite{ngo2017cell}. However, to date, most studies on cell-free massive MIMO do not formally model this type of impairment, so, strictly speaking, they  do not perform truly distributed beamforming optimization. To address this limitation, as in \cite{miretti2021team}, we model constraints related to imperfect CSI sharing with the following set:
\begin{equation}\label{eq:distributedCSI}
\set{V}_k = \set{H}_1^N \times \ldots \times \set{H}_L^N,
\end{equation}
where $\set{H}_l^N\subseteq \set{H}^N$ denotes the set of $N$-tuples of $\Sigma_l$-measurable functions $\Omega \to \stdset{C}$ satisfying $(\forall \rvec{h}\in \set{H}_l^N)$ $\E[\|\rvec{h}\|_2^2]<\infty$, and where $\Sigma_l \subseteq \Sigma$ is the sub-$\sigma$-algebra induced by the available CSI at the $l$th AP, also called the \emph{information subfield} of AP $l$ \cite{yukselbook}. Furthermore, another practical limitation is that only a subset of APs are used for decoding the message of a given UE. More specifically, scalable cell-free networks are expected to implement some type of user-centric rule \cite{buzzi2020}, which allocates to each UE a cluster of serving APs. As shown in \cite{miretti2021team2}, this practical limitation can be modeled by setting $\set{H}_l^N$ in \eqref{eq:distributedCSI} to the set $\{\vec{0}_N~ (\mathrm{a.s.})\}$ if AP $l$ is not serving UE $k$. Considering both impairments, we remark that $\set{V}_k$ is a linear \emph{subspace} of $\set{H}^{NL}$.  

\subsection{Weighted max-min SINR optimization}
The objective of the proposed algorithms is to optimize jointly the distributed beamformers and the transmit powers, with the intent to obtain rate tuples on the boundary $\partial\set{R}$ of $\set{R}$ (see Definition~\ref{def:pareto}). More precisely, under consideration is the following optimization problem:
\begin{equation}\label{eq:maxmin}
\begin{aligned}
\underset{\substack{\vec{p}\in \stdset{R}_{++}^K \\ \rvec{v}_1,\ldots,\rvec{v}_K \in \set{H}^{NL} }}{\text{maximize}}
& \quad  \min_{k\in\set{K}} \omega_k^{-1} \mathrm{SINR}_k(\rvec{v}_k,\vec{p}) \\
\text{subject to} & \quad \|\vec{p}\|_{\infty}\leq P \\
 & \quad  (\forall k \in \set{K})~\rvec{v}_k \in \set{V}_k\\
 & \quad  (\forall k \in \set{K})~\E[\|\rvec{v}_k\|_2^2] \neq 0,
\end{aligned}
\end{equation}
where $(\omega_1,\ldots,\omega_K)\in \stdset{R}_{++}^K$ is a given vector of positive weights, which can be tuned to focus on different operational points in $\partial\set{R}$. In particular, the popular \emph{max-min fair} point is obtained by setting $\vec{\omega} = \vec{1}$.
More generally, for different choices of weights, we have the next result:
\begin{proposition}\label{prop:pareto}
If a solution $(\vec{p}^\star,\rvec{v}_1^\star,\ldots,\rvec{v}_K^\star)$ to Problem~\eqref{eq:maxmin} exists, then $\vec{r}^\star \eqdef (R_1(\rvec{v}_1^\star,\vec{p}^\star),\ldots,R_K(\rvec{v}_K^\star,\vec{p}^\star))$ belongs to $\partial \set{R}$. Furthermore, every positive point in $\partial \set{R}$ can be obtained for some $\vec{\omega}\in \stdset{R}_{++}^K$.
\end{proposition}  
\begin{proof} 
(Sketch) The proof follows from standard properties of performance regions in multi-criteria resource allocation  problems (see, e.g, \cite{emil2013resource}), and the logarithmic relation between rate and SINR. The details are omitted. 
\end{proof}

\begin{remark}
A popular alternative to weighted max-min optimization is \emph{weighted sum-rate} (WSR) optimization (see, e.g., \cite{demir2021}). One major difference between these two approaches is that optimal WSR solutions always produce points on the \emph{strong} Pareto boundary of $\set{R}$ (see \cite{emil2013resource} for a definition). However, unlike weighted max-min solutions, if $\set{R}$ is nonconvex, WSR solutions may miss some (weak or strong) Pareto optimal points that are important for a given criterion of fairness \cite{das1997closer}, \cite[Remark 1.3]{emil2013resource}. Nonconvex rate regions are common, for instance, when the network enters an interference limited regime. User fairness can be restored through \emph{scheduling}, but resource allocation becomes more complex \cite{emil2013resource}.
\end{remark}  

\section{Optimal iterative solution}
In this section we present a fast iterative algorithm that provably converges to a solution to Problem~\eqref{eq:maxmin}, provided that a solution exists and that it attains a positive optimum. As we prove later, existence and positivity of a solution is guaranteed by a simple condition on the fading and CSI distribution. Specifically, it is guaranteed if $(\forall k \in \set{K})~\set{V}_k':=\{\rvec{v}_k \in \set{V}_k~|~\E[\rvec{h}_k^\herm \rvec{v}_k]\neq 0\}$ is nonempty, which holds for most fading and CSI distributions.
\begin{remark}
In particular, $\set{V}_k'\neq \emptyset$ does not hold if the APs have no CSI and, in addition, the channels have zero mean. In this case, $(\forall \rvec{v}_k \in \set{V}_k)~\E[\rvec{h}_k^\herm\rvec{v}_k] = \E[\rvec{h}_k]^\herm\E[\rvec{v}_k] = 0$ holds. This is one of the main drawbacks of the UatF bounding technique, which requires either sufficient CSI \emph{or} a strong channel mean to produce non-trivial capacity inner bounds.
\end{remark}

\subsection{Power control using interference calculus}
As a preliminary step, we first revisit known results for fairly general power control problems of the type 
\begin{equation}\label{eq:power_control}
\underset{\|\vec{p}\|\leq P}{\text{maximize}} \min_{k\in\set{K}} \omega_k^{-1} u_k(\vec{p}),
\end{equation}
where $\|\cdot\|$ is a monotone norm on $\stdset{R}^K$, and where $u_k:\stdset{R}_+^K\to \stdset{R}_+$ denotes the utility of UE $k$ satisfying
\begin{equation}\label{eq:interf}
(\forall\vec{p}\in\stdset{R}_+^K)~u_k(\vec{p}) = \dfrac{p_k}{f_k(\vec{p})}
\end{equation}
for some standard interference function $f_k:\stdset{R}_+^K\to\stdset{R}_{++}$ (see Definition \ref{def:interf_fun}). This structure covers many common power control problems in wireless networks, including an equivalent formulation of Problem~\eqref{eq:maxmin}. An important consequence of the above structure is the following result:
\begin{proposition}\label{prop:fixed_point}
Define $T:\stdset{R}_+^K\to \stdset{R}_{++}^K:\vec{p}\mapsto [\omega_1f_1(\vec{p}), \ldots,\omega_Kf_K(\vec{p})]^\T$. Then each of the following holds:
\begin{enumerate}[(i)]
\item The set $\mathrm{Fix}(\tilde{T})$ of fixed points of the mapping 
\begin{equation*}
\tilde{T}:  \stdset{R}_{+}^K \to \stdset{R}_{++}^K  : \vec{p} \mapsto \frac{P}{\|T(\vec{p})\|}T(\vec{p})
\end{equation*}
is a singleton, and  its unique fixed point $\vec{p}^\star\in\mathrm{Fix}(\tilde{T})$  is a solution to Problem~\eqref{eq:power_control}.
\item For every $\vec{p}_1 \in \stdset{R}_{+}^K$, the sequence $(\vec{p}_i)_{i\in\stdset{N}}$ generated via 
\begin{equation}
	\label{eq:italg}
(\forall i \in \stdset{N})~\vec{p}_{i+1} = \tilde{T}(\vec{p}_i)
\end{equation}
converges in norm to the fixed point $\vec{p}^\star\in\mathrm{Fix}(\tilde{T})$.
\item Denote by $\set{P}\subset\stdset{R}_{++}^K$ the nonempty set of solutions to Problem~\eqref{eq:power_control}. Then the fixed point $\vec{p}^\star\in\mathrm{Fix}(\tilde{T})$ satisfies $(\forall \vec{p}\in \set{P})~\vec{p}\ge \vec{p}^\star$; i.e., $\vec{p}^\star$ is the least element of $\set{P}$ w.r.t. the partial ordering induced by the nonnegative cone.
\end{enumerate}
\end{proposition}
\begin{proof}
	(Sketch) (i) The proof that the mapping $\tilde{T}$ has a unique fixed point follows directly from \cite[Theorem 3.2]{nuzman07}. We can use arguments similar to those in \cite[Sect.~III]{zheng2016} to show that the fixed point of $\tilde{T}$ is a solution to Problem~\ref{eq:power_control}. However, some care has to be taken in following those arguments because, unlike the results in \cite{zheng2016}, utilities with the structure in \eqref{eq:interf} do not necessarily satisfy all properties in \cite[Assumption~1]{zheng2016}, so, in particular, the solution to Problem~\eqref{eq:power_control} is not necessarily unique. We omit the full proof because of the space limitation.
	
	(ii) Use the same arguments shown in the proof of \cite[Theorem 3.2]{nuzman07}.
	
	(iii) Omitted because of the space limitation.
\end{proof}

\begin{remark}
From a practical perspective, Proposition~\ref{prop:fixed_point} shows that Problem~\eqref{eq:power_control} always has a solution. Furthermore, if the solution is not unique, then the recursion in \eqref{eq:italg} converges in norm to the solution with minimum sum transmit power.
\end{remark}


\subsection{Equivalent problem reformulation}
We now map Problem~\eqref{eq:maxmin} to an equivalent power control problem of the type in \eqref{eq:power_control}, where the optimization of the beamformers is implicit in the definition of the utilities $u_k$. More precisely, by exploiting the fact that the SINRs are only coupled through the power vector $\vec{p}$, we obtain:
\begin{proposition}\label{prop:equivalence}
For each UE $k\in \set{K}$, define the utility
\begin{equation}\label{eq:maxSINR}
(\forall\vec{p}\in\stdset{R}_+^K)~u_k(\vec{p}) \eqdef \sup_{\substack{\rvec{v}_k\in\set{V}_k\\ \E[\|\rvec{v}_k\|_2^2] \neq 0}}\mathrm{SINR}_k(\rvec{v}_k,\vec{p}),
\end{equation}
and assume that $(\forall k~\in \set{K})~\set{V}'_k \eqdef \{\rvec{v}_k \in \set{V}_k~|~\E[\rvec{h}_k^\herm \rvec{v}_k]\neq 0\} \neq \emptyset$. Then, each of the following holds: 
\begin{enumerate}[(i)]
\item The utility \eqref{eq:maxSINR} satisfies $(\forall k \in \set{K})(\forall\vec{p}\in\stdset{R}_{++}^K)$ 
\begin{equation*}
u_k(\vec{p}) = \sup_{\rvec{v}_k\in\set{V}'_k}\mathrm{SINR}_k(\rvec{v}_k,\vec{p}) > 0;
\end{equation*}
\item The utility \eqref{eq:maxSINR} satisfies \eqref{eq:interf} with the standard interference function $(\forall k \in \set{K})(\forall\vec{p}\in\stdset{R}_+^K)~f_k(\vec{p}) \eqdef$
\begin{equation*}
\inf_{\rvec{v}_k \in \set{V}'_k} \dfrac{p_k\Var(\rvec{h}_k^\herm\rvec{v}_k)+\sum_{j\neq k}p_j\E[|\rvec{h}_j^\herm\rvec{v}_k|^2]+\E[\|\rvec{v}_k\|_2^2]}{|\E[\rvec{h}_k^\herm\rvec{v}_k]|^2};
\end{equation*}
\item Let $\vec{p}^\star$ be a solution to Problem~\eqref{eq:power_control} with utilities in \eqref{eq:maxSINR} and monotone norm $\|\cdot\| = \|\cdot\|_{\infty}$. If $(\forall k~\in \set{K})$ $\exists \rvec{v}^\star_k\in \set{V}'_k$ such that $u_k(\vec{p}^\star) = \mathrm{SINR}_k(\rvec{v}^\star_k,\vec{p}^\star)$ (i.e., the $\sup$ is attained), then $(\vec{p}^\star,\rvec{v}_1^\star,\ldots,\rvec{v}_K^\star)$ solves Problem~\eqref{eq:maxmin}.
\end{enumerate} 
\end{proposition}
\begin{proof}
\begin{enumerate}[(i)]
\item Since $\set{V}'_k  \subset \{\rvec{v}_k \in \set{V}_k~|~\E[\|\rvec{v}_k\|_2^2] \neq 0\}$, we have $(\forall\vec{p}\in\stdset{R}_+^K)~u_k(\vec{p}) \geq \sup_{\rvec{v}_k\in\set{V}'_k}\mathrm{SINR}_k(\rvec{v}_k,\vec{p})$. By definition of $\set{V}'_k$, $\sup_{\rvec{v}_k\in\set{V}'_k}\mathrm{SINR}_k(\rvec{v}_k,\vec{p}) > 0$ for $p_k>0$. This also shows that beamformers such that $\E[\rvec{h}_k^\herm \rvec{v}_k]= 0$ can be omitted when evaluating \eqref{eq:maxSINR}.
\item To confirm that $f_k$ is indeed a standard interference function, recall that $f_k$ is concave because it is constructed by taking the point-wise infimum of affine (and hence concave) functions. Furthermore, by definition of $\mathcal{V}_k'$, which also implies $(\forall \rvec{v}_k\in \mathcal{V}_k')~\E[\|\rvec{v}_k\|_2^2] \neq 0$, 
\begin{align*}
(\forall\vec{p}\in\stdset{R}_+^K)~f_k(\vec{p})\geq \inf_{\rvec{v}_k \in \set{V}'_k}\dfrac{\E[\|\rvec{v}_k\|_2^2]}{|\E[\rvec{h}_k^\herm\rvec{v}_k]|^2} >0.
\end{align*} 
Therefore, we conclude that $f_k$ is not only concave but also positive, and hence it is a standard interference function \cite[Proposition~1]{renato2016}.
\item Immediate, hence omitted. 
\end{enumerate}
\end{proof}
The main implication of Proposition~\ref{prop:equivalence} is that an optimal power vector $\vec{p}^\star$ for Problem~\eqref{eq:maxmin} can be obtained via the fixed-point iterations in \eqref{eq:italg}, where each $i$th step involves the evaluation of the standard interference mapping $T(\vec{p})= [\frac{\omega_1p_1}{u_1(\vec{p})}, \ldots,\frac{\omega_Kp_K}{u_K(\vec{p})}]^\T$ for $\vec{p}=\vec{p}_i$, and where $u_k(\vec{p})$ is given by \eqref{eq:maxSINR}. The main challenge of this approach lies in the computation of $u_k(\vec{p})$, which requires solving a SINR maximization problem over the set of feasible distributed beamformers. The following section discusses how to perform this step, and it shows that the supremum in \eqref{eq:maxSINR} is always attained in practice.

\subsection{Team MMSE beamforming}
The following proposition states that optimal distributed beamformers under fixed power control, i.e., beamformers that attain the supremum in \eqref{eq:maxSINR} for a given $\vec{p}$, can be obtained by solving a minimum mean squared error (MMSE) problem under information constraints. For convenience, we define the random channel matrix $\rvec{H}\eqdef[\rvec{h}_1,\ldots,\rvec{h}_K]$.
\begin{proposition}
\label{prop:MSE}
For given $k\in\set{K}$ and $\vec{p}\in\stdset{R}_{++}^K$, consider the optimization problem 
\begin{equation}\label{eq:MSE}
\underset{\rvec{v}_k \in \set{V}_k}{\emph{minimize}}~\mathrm{MSE}_k(\rvec{v}_k,\vec{p}),
\end{equation}
where we define $(\forall \rvec{v}_k \in \set{H}^{NL})$
\begin{equation*}
\mathrm{MSE}_k(\rvec{v}_k,\vec{p})\eqdef \E\left[\|\mathrm{diag}(\vec{p})^{\frac{1}{2}}\rvec{H}^\herm\rvec{v}_k-\vec{e}_k\|_2^2\right] + \E\left[\|\rvec{v}_k\|_2^2\right].
\end{equation*}
Problem~\eqref{eq:MSE} has a unique solution $\rvec{v}_k^\star\in \set{V}_k$. Furthermore, if $\set{V}_k'\neq \emptyset$, then $u_k(\vec{p}) = \mathrm{SINR}_k\left(\rvec{v}_k^\star,\vec{p}\right)$ holds.
\end{proposition}
\begin{proof} 
Existence and uniqueness of the solution  follows from \cite[Lemma 4]{miretti2021team}. For the second statement, we observe that  
\begin{align*}
&\inf_{\rvec{v}_k\in\mathcal{V}_k}\mathrm{MSE}_k(\rvec{v}_k,\vec{p}) \\
&\overset{(a)}{=} \inf_{\substack{\rvec{v}_k\in\set{V}_k\\ \E[\|\rvec{v}_k\|_2^2] \neq 0}} \inf_{\beta \in \stdset{C}} \mathrm{MSE}_k(\beta\rvec{v}_k,\vec{p}) \\
&\overset{(b)}{=} \inf_{\substack{\rvec{v}_k\in\set{V}_k\\ \E[\|\rvec{v}_k\|_2^2] \neq 0}} 1-\dfrac{p_k|\E[\rvec{h}_k^\herm\rvec{v}_k]|^2}{\sum_{j\in\set{K}}p_j\E[|\rvec{h}_j^\herm\rvec{v}_k|^2]+\E\left[\|\rvec{v}_k\|_2^2\right]} \\
&=  \inf_{\substack{\rvec{v}_k\in\set{V}_k\\ \E[\|\rvec{v}_k\|_2^2] \neq 0}}\dfrac{1}{1+\mathrm{SINR}_k(\rvec{v}_k,\vec{p})} = \dfrac{1}{1+u_k(\vec{p})} \overset{(c)}{<} 1,
\end{align*}
where $(a)$ follows from $(\forall \beta\in \stdset{C})(\forall \rvec{v}_k \in \mathcal{V}_k)$ $\beta \rvec{v}_k \in\mathcal{V}_k$, $(b)$ follows from standard minimization of scalar quadratic forms as in the proof of the UatF bound (see, e.g., \cite{marzetta2016fundamentals,massivemimobook,miretti2021team}), and $(c)$ follows from $\set{V}_k'\neq\emptyset$ and Proposition~\ref{prop:equivalence}(i). If $\rvec{v}_k^\star \neq \vec{0}$ a.s., the identity $u_k(\vec{p}) = \mathrm{SINR}_k\left(\rvec{v}_k^\star,\vec{p}\right)$ follows readily from the above chain of equalities. The proof is concluded by observing that $\inf_{\rvec{v}_k\in\mathcal{V}_k}\mathrm{MSE}_k(\rvec{v}_k,\vec{p})<1$ implies $\rvec{v}_k^\star \neq \vec{0}$ a.s..  
\end{proof}
The solution to \eqref{eq:MSE}, which can be interpreted as the best distributed approximation of a regularized zero-forcing beamformer, can be obtained via the recently developed \emph{team} MMSE beamforming framework given by \cite{miretti2021team}. Although presented in the context of downlink beamforming with perfect message sharing, the method in \cite{miretti2021team} directly applies to the uplink case and to user-centric clustering \cite{miretti2021team2}. This method states that the solution to \eqref{eq:MSE} corresponds to the unique solution to a linear feasibility problem, which, depending on the CSI structure, can be solved explicitly or approximately via efficient numerical techniques. We remark that this method can be applied to fairly general CSI structures, for instance involving \textit{quantized} or \textit{delayed} information sharing, or exploiting the peculiarities of efficient fronthauls such as in the so-called \emph{radio stripe} concept \cite{miretti2021team}. 

To keep the proposed algorithms general, in the following section we do not specify the setup (i.e., $\rvec{H}$ and the information structure $\set{V}_1,\ldots,\set{V}_K$), and we simply assume that a solution to \eqref{eq:MSE} is available. However, as an example, our numerical results consider the particular case of local CSI as in  \cite{ngo2017cell}, together with a simple user-centric clustering rule. 


\subsection{Proposed algorithms}
We now have all the elements for deriving an algorithmic solution to Problem~\eqref{eq:maxmin}, i.e., to perform joint optimal long-term power control and distributed beamforming design. Specifically, we propose to solve Problem~\eqref{eq:maxmin} via the fixed-point iterations in \eqref{eq:italg}, as summarized below:
\begin{algorithmic}
\Require $\vec{p} \in \stdset{R}_{++}^K$
\Repeat
\For{$k\in\set{K}$}
\State $\rvec{v}_k \gets  \arg\min_{\rvec{v}_k\in\set{V}_k}\mathrm{MSE}_k(\rvec{v}_k,\vec{p})$ 
\EndFor
\State $\vec{t} \gets \begin{bmatrix}
\frac{\omega_1p_1}{\mathrm{SINR}_1(\rvec{v}_1,\vec{p})} & \ldots & \frac{\omega_Kp_K}{\mathrm{SINR}_K(\rvec{v}_K,\vec{p})}
\end{bmatrix}^\T$
\State $\vec{p}\gets \frac{P}{\|\vec{t}\|_{\infty}}\vec{t}$
\Until{no significant progress is observed.}
\end{algorithmic}
Asymptotic convergence to an optimal solution is guaranteed by Proposition~\ref{prop:fixed_point}, Proposition~\ref{prop:equivalence}, and Proposition~\ref{prop:MSE}. Clearly, the above algorithm relies on solving Problem~\eqref{eq:MSE} at each iteration, which, as already mentioned, can be done following \cite{miretti2021team,miretti2021team2}. Note that $\rvec{v}_k$ is a \textit{function} of the CSI, and not a deterministic vector. In practice, as we will see in the numerical examples, storing $\rvec{v}_k$ during the algorithm execution means storing its long-term parameters (e.g., the regularization factor in a regularized channel inversion block).  

An intuitive variant of the first algorithm proposed above alternates between distributed beamforming optimization and power control, and it is summarized below:
\begin{algorithmic}
\Require $\vec{p} \in \stdset{R}_{++}^K$
\Repeat
\For{$k\in\set{K}$}
\State $\rvec{v}_k \gets  \arg\min_{\rvec{v}_k\in\set{V}_k}\mathrm{MSE}_k(\rvec{v}_k,\vec{p})$ 
\EndFor
\State $\vec{p}\gets\arg\max_{\|\vec{p}\|_\infty \leq P}  \min_{k\in\set{K}} \omega_k^{-1} \mathrm{SINR}_k(\rvec{v}_k,\vec{p})$
\Until{no significant progress is observed.}
\end{algorithmic}
In addition to solving Problem~\eqref{eq:MSE}, this algorithm requires a solution to $\max_{\|\vec{p}\|_\infty \leq P}  \min_{k\in\set{K}} \omega_k^{-1} \mathrm{SINR}_k(\rvec{v}_k,\vec{p})$ for given beamformers $(\rvec{v}_1,\ldots,\rvec{v}_K)$, which is obtained in closed form in \cite{miretti2022power}. The advantage is that this power control step maximizes the improvement in utility. The disadvantage is that performing this power control step is more complex than performing a fixed point iteration. However, this increase in complexity is in practice negligible when compared to the complexity of solving Problem~\eqref{eq:MSE}. Convergence properties of this algorithm will be studied in a future work, and they are here only tested using numerical experiments.

\section{Numerical experiments and final remarks}
\subsection{Simulation setup}
\begin{figure}
\centering
\includegraphics[width=0.8\columnwidth]{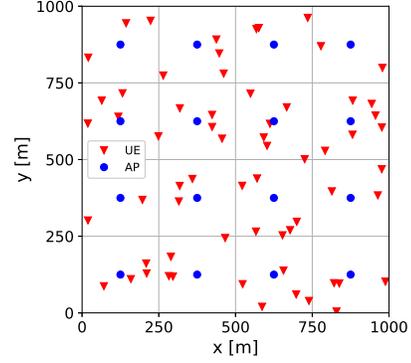}
\caption{Pictorial representation of the simulated setup: $K=64$ UEs uniformly distributed within a squared service area of size $1\times 1~\text{km}^2$, and $L=16$ regularly spaced APs with $N=8$ antennas each. Each UE is jointly served by a cluster of $Q=4$ APs offering the strongest channel gains.}
\label{fig:network}
\end{figure}
We consider the network depicted in Figure~\ref{fig:network}, where $K=64$ UEs are uniformly distributed within a squared service area of size $1\times 1~\text{km}^2$, and $L=16$ regularly spaced APs with $N=8$ antennas each. By neglecting for simplicity channel correlation, we let each sub-vector $\rvec{h}_{l,k}$ of $\rvec{h}_k^{\herm} =: [\rvec{h}_{1,k}^\herm, \ldots \rvec{h}_{L,k}^\herm]$ be independently distributed as $\rvec{h}_{l,k} \sim \CN\left(\vec{0}, \gamma_{l,k}\vec{I}_N\right)$, where $\gamma_{l,k}>0$ denotes the channel gain between AP $l$ and UE $k$. We follow the same 3GPP-like path-loss model adopted in \cite{demir2021} for a $2$ GHz carrier frequency:
\begin{equation*}
\gamma_{l,k} = -36.7 \log_{10}\left(D_{l,k}/1 \; \mathrm{m}\right) -30.5 + Z_{l,k} -\sigma^2 \quad \text{[dB]},
\end{equation*}
where $D_{l,k}$ is the distance between AP $l$ and UE $k$ including a difference in height of $10$ m, and $Z_{l,k}\sim \CN(0,\rho^2)$ [dB] are shadow fading terms with deviation $\rho = 4$. The shadow fading is correlated as $\E[Z_{l,k}Z_{j,i}]=\rho^22^{-\frac{\delta_{k,i}}{9 \text{ [m]}}}$ for all $l=j$ and zero otherwise, where $\delta_{k,i}$ is the distance between UE $k$ and UE $i$. The noise power is 
$\sigma^2 = -174 + 10 \log_{10}(B) + F$ [dBm],
where $B = 20$ MHz is the bandwidth, and $F = 7$ dB is the noise figure. The  power budget is set to $P = 20$ dBm.

\subsection{Local CSI and path loss based user-centric clustering}
We solve Problem~\eqref{eq:maxmin} by focusing on the following simple user-centric distributed information structure with no short-term CSI sharing. We assume that each UE $k$ is served only by its $Q=4$ strongest APs, i.e., by the subset of APs indexed by $\set{L}_k\subseteq \set{L}$, where each set $\set{L}_k$ is formed by ordering $\set{L}$ w.r.t. decreasing $\gamma_{l,k}$ and by keeping only the first $Q$ elements. Then, we assume each AP $l$ to acquire local CSI $\hat{\rvec{H}}_l\eqdef[\hat{\rvec{h}}_{l,1},\ldots,\hat{\rvec{h}}_{l,K}]$,  
\begin{equation*}
(\forall k \in\set{K})~\hat{\rvec{h}}_{l,k} \eqdef \begin{cases}
\rvec{h}_{l,k} & \text{if } l \in \set{L}_k,\\
\E[\rvec{h}_{l,k}] & \text{otherwise}.
\end{cases}
\end{equation*}
This model mirrors the fully distributed implementation proposed in \cite{bjornson2020scalable}. For simplicity, we neglect estimation noise, and we focus on a model where small-scale fading coefficients are either perfectly known at some APs or completely unknown. Under this model, the optimal solution $[\rvec{v}^\herm_{1,k},\ldots,\rvec{v}^\herm_{L,k}] \eqdef \rvec{v}^\herm_k \in \set{V}_k$ to Problem~\eqref{eq:MSE} is given by \cite{miretti2021team,miretti2021team2}
\begin{equation}\label{eq:TMMSE}
(l\in \set{L})(k\in\set{K})~\rvec{v}_{l,k} = \rvec{V}_l\vec{c}_{l,k}, 
\end{equation}
where: 
\begin{itemize}
\item $\rvec{V}_l \eqdef \left(\hat{\rvec{H}}_l\mathrm{diag}(\vec{p})\hat{\rvec{H}}_l^\herm + \vec{\Psi}_l\right)^{-1}\hat{\rvec{H}}_l\mathrm{diag}(\vec{p})^{\frac{1}{2}}$ is a \textit{local} MMSE  beamforming stage \cite{bjornson2020scalable} with augmented noise covariance $\vec{\Psi}_l \eqdef (1+\sum_{\{k\in \set{K}|l\notin\set{L}_k\}}p_k\gamma_{l,k})\vec{I}_N$;
\item $\vec{c}_{l,k}\in \stdset{C}^K$ is a statistical beamforming stage given by the unique solution to the linear system of equations
\begin{equation*}
\begin{cases}\vec{c}_{l,k} + \sum_{j \in \set{L}_k \backslash \{l\}}\vec{\Pi}_j \vec{c}_{j,k} = \vec{e}_k & \forall l \in \set{L}_k, \\
\vec{c}_{l,k} = \vec{0}_{K\times 1} & \text{otherwise,}
\end{cases}
\end{equation*}
where $\vec{\Pi}_l \eqdef \E\left[\mathrm{diag}(\vec{p})^{\frac{1}{2}}\hat{\rmat{H}}_l^\herm\rmat{V}_l\right]$.
\end{itemize} 
The main difference between \eqref{eq:TMMSE} and the best known heuristic, called local MMSE beamforming with optimal \textit{large-scale fading decoding} (see, e.g., \cite{demir2021}), is that the latter considers only a scalar statistical stage, i.e., only one component of $\vec{c}_{l,k}$ is different than zero. We refer to \cite{miretti2021team} for additional comparisons against known heuristics, and for the extension to noisy CSI. 

\subsection{Illustration of the proposed algorithms}
\begin{figure}
\centering
\includegraphics[width=0.9\columnwidth]{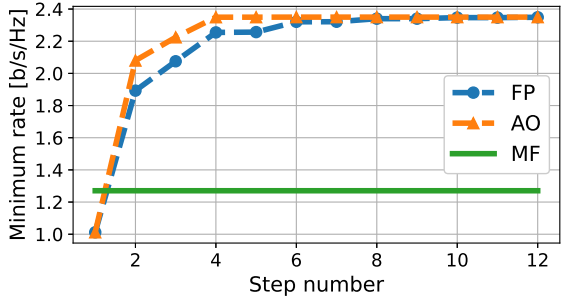}
\caption{Minimum achievable rate for a given step of the proposed fixed point (FP) and alternate optimization (AO) algorithms. The odd steps perform optimal distributed beamforming design, while the even steps update the long-term power control coefficients. Both algorithms attain the optimal max-min fair point with few iterations. As a baseline, we plot the performance of classical matched filter (MF) beamforming with optimal power control \cite{demir2021}.}
\label{fig:distributed}
\end{figure}
Figure~\ref{fig:distributed} illustrates several steps of the proposed algorithms, applied to the computation of the optimal \textit{max-min} fair point of $\set{R}$, i.e., the solution to Problem~\eqref{eq:maxmin} for $\vec{\omega} = \vec{1}$. The initialization is the full power case $\vec{p}=P\vec{1}$, which turns out to be highly suboptimal. Both algorithms require very few iterations for approaching the optimum, and, not surprisingly, alternate optimization  outperforms fixed point iterations, at the price of a slight increase in complexity. We point out that, although not shown in this paper, fast convergence has been observed also under many different simulation parameters. 

It is worthwhile underlying that the implemented algorithms operate only by updating deterministic coefficients using long-term information. In particular, when it comes to optimal distributed beamforming design, the implemented routines update only the long-term parameters $\vec{c}_{l,k}$ and $\vec{\Psi}_l$ in \eqref{eq:TMMSE}, based on a fixed training set of $N_{\mathrm{sim}}=1000$ channel realizations and channel gain information. 


\bibliographystyle{IEEEbib}
\bibliography{IEEEabrv,refs}

\end{document}